\documentclass{article}
\usepackage{spconf}

\usepackage{color}
\usepackage{bm}
\usepackage{bbm}
\usepackage{amsmath}
\usepackage{amssymb}
\usepackage{graphicx}
\usepackage{mathrsfs}
\usepackage{empheq}
\usepackage[mathcal]{euscript}
\usepackage[margin = 2cm]{geometry}
\usepackage{amsthm}
\usepackage{xcolor}
\usepackage{cite}
\usepackage{balance}

\newtheorem{theorem}{Theorem}
\newtheorem{lemma}{Lemma}

\newtheorem{remark}{Remark}
\newtheorem{assumption}{Assumption}

\newcommand{\beq}{\begin{equation}}
\newcommand{\eeq}{\end{equation}}
\newcommand{\beqa}{\begin{eqnarray}}
\newcommand{\eeqa}{\end{eqnarray}}	
\newcommand{\thetashare}{\theta_{\sf TX}}
\newcommand{\thetatrue}{\theta_0}
\newcommand{\T}{^\top}
\allowdisplaybreaks

\newcommand*{\QED}{\hfill\ensuremath{\square}}%

\title{Social Learning with Partial Information Sharing}
\name{Virginia Bordignon$^{\star}$ \qquad Vincenzo Matta$^{\dagger}$ \qquad Ali H. Sayed$^{\star}$\thanks{This work was supported in part by the Swiss National Science Foundation under grant 205121-184999.
 E-mails: virginia.bordignon@epfl.ch, vmatta@unisa.it, ali.sayed@epfl.ch}}

\address{$^{\star}$ School of Engineering, Ecole Polytechnique F\'ed\'erale de Lausanne (EPFL)\\$^{\dagger}$DIEM, University of Salerno}

\begin{document}
\ninept
\maketitle
\begin{abstract}
This work studies the learning abilities of agents sharing partial beliefs over social networks. The agents observe data that could have risen from one of several hypotheses and interact locally to decide whether the observations they are receiving have risen from a particular hypothesis of interest. To do so, we establish the conditions under which it is sufficient to share {\em partial} information about the agents' belief in relation to the hypothesis of interest. Some interesting convergence regimes arise.
\end{abstract}
\vspace{-5pt}
\section{Introduction and Related Work}
\vspace{-5pt}
\label{sec:intro}
Modeling opinion formation over social networks is a subject of great interest, including in modern times with the proliferation of social platforms. Many algorithmic approaches have been conceived for this purpose \cite{chamley2013models, acemoglu2011bayesian,molavi2018theory, friedkin2016network}, including the non-Bayesian approach, in which agents update their beliefs or opinions by using local streaming observations and by combining information shared by their neighbors. Some of the main studies along these lines rely on consensus and diffusion strategies \cite{degroot1974reaching,sayed2013diffusion}, both with linear and log-exponential belief combinations (see, e.g., \cite{jadbabaie2012non, zhao2012learning, shahrampour2015distributed,salami2017social,nedic2017fast,lalitha2018social}). In all of these works, it is assumed that agents share with their neighbors their {\em entire} belief vectors. 

In this work, we consider that agents share only {\em partial} beliefs. We motivate the model by means of an example. Consider a collection of agents connected by a topology and monitoring the weather conditions in a certain region of space. Each agent is measuring some local data that is dependent on the weather conditions, say, temperature, pressure, wind speed and direction. The agents wish to interact to decide whether the state of nature is one of three possibilities: sunny, rainy or snowing. At every time instant $i$, each agent $k$ will have a belief vector, say, $\mu_{k,i}(\theta)=[0.6, 0.2, 0.2]$. These values mean that, at time $i$, agent $k$ believes it is sunny with probability 0.6.  Similarly, all other agents in the network will have their own beliefs. If the weather conditions are sunny, we would like the agents to converge after sufficient iterations to the belief vector $[1,0,0]$. This type of convergence behavior has already been established in prior studies \cite{jadbabaie2012non, lalitha2018social} for strongly-connected networks (i.e., over networks where there is always a path and a reverse path between any two agents, in addition to at least one agent having a self-loop to reflect confidence in their local data). 

But what if we now formulate a different question for the network to answer? Assume the objective is for the network to decide ``whether it is sunny". Will the agents still need to share their {\em entire} belief vectors repeatedly to find out whether it is sunny or not? What if we devise a cooperation scheme where agents share only, at every iteration, the likelihood they have for the ``sunny" condition and ignore the other entries in their belief vector. Will they still be able to learn the state of nature? That is, can agents still learn if they only share partial information without needing to share simultaneously information about {\em all} possible hypotheses?

This work answers this question and provides conditions under which different interesting forms of convergence regimes can occur. To begin with, we present in Eqs. \eqref{eq:Bayesupdate}--\eqref{eq:combine_nc} below the learning algorithm by which agents share partial information in relation to the particular question that the network is trying to answer. We will then show that if the presumed hypothesis is the true one, then the network will be able to learn correctly even under the regime of partial information sharing. Otherwise, if there is sufficient Kullback-Leibler (KL) divergence between the true likelihood and the likelihood under investigation, then the network will be able to correctly discount the latter. However, and as would be expected, a mislearning scenario is possible if the KL divergence between true and presumed hypotheses is small enough, in which case the presumed hypothesis is confounded with the true one. The technical details, which are more involved than this brief intuitive summary, are spelled out in the sequel with illustrative examples and supporting simulations. Due to space limitations some proofs are omitted.

Theoretical results are provided for two alternatives: $i)$ the {\em partial} approach, where agents are so focused on the partial shared opinion that they segment also their own belief; $ii)$  and the \emph{self-aware partial} approach where each agent takes into account its own true belief vector. The analysis of both algorithms points toward interesting behavior, which complements our understanding of social learning systems. It also motivates the enunciation of a third algorithm, which will be merely introduced in this work: it consists of agents sharing their most relevant belief component at each time.

\vspace{-5pt}
\section{Partial information}
\vspace{-5pt}
\label{sec:slpartial}
In this work, we consider a strongly connected network of $N$ agents interacting to learn 
an unknown state of nature $\theta\in\Theta$. Each agent $k$ observes, at instant $i$, streaming data $\bm{\xi}_{k,i}\in\mathcal{X}$ (we use bold notation for random variables), which can be either continuous or discrete random variables assumed independent across time. 
The likelihood of the data corresponding to a given hypothesis $\theta$ will be denoted by $L(\xi|\theta)$, for $\xi\in\mathcal{X}$. 
The data $\bm{\xi}_{k,i}$ are distributed according to $L(\xi|\theta_0)$, where $\theta_0$ is the true hypothesis. 


We model the topology by means of a strongly-connected graph with a left-stochastic combination matrix $A$ whose nonnegative entries $a_{\ell k}$ are equal to zero if $\ell \notin \mathcal{N}_k$, with $\mathcal{N}_k$ being the set of neighbors of agent $k$. We denote by $v$ the Perron eigenvector of $A$, which is defined as
\begin{equation}
	Av=v,\qquad \mathbbm{1}\T v=1, \qquad v \succ 0.
\end{equation}
\begin{assumption}[Positive initial beliefs]\label{ass:2}
	In the absence of any prior information, all agents assign positive initial beliefs to all hypotheses, i.e., $\bm{\mu}_{k,0}(\theta)>0$ for each agent $k$ and all $\theta\in \Theta$.\QED
\end{assumption}
\begin{assumption}[Likelihood ratios]\label{ass:3}
	For each agent $k$ and $i=1,2,\dots$, the random variables $\log \left(L(\bm{\xi}_{k,i}|\theta)/L(\bm{\xi}_{k,i}|\theta')\right)$ are integrable for every pair $\{\theta, \theta'\}\in \Theta$.\QED
\end{assumption}
\subsection{Partial approach}
Let $\thetashare$ denote a hypothesis of interest\footnote{The subscript {\sf TX} denotes the ``transmitted" hypothesis.}.\,\,The objective for the agents is to verify whether the state of nature agrees with $\thetashare$ or not. Thus, we propose the following modified version of the log-linear social learning algorithm used in \cite{lalitha2018social,matta2019exponential}:
\begin{align}
\bm{\psi}_{k,i}(\theta)&=
\displaystyle{
	\frac
	{\bm{\mu}_{k,i-1}(\theta)L(\bm{\xi}_{k,i} | \theta)}
	{\sum\limits_{\theta^{\prime}\in\Theta} \bm{\mu}_{k,i-1}(\theta^{\prime})L(\bm{\xi}_{k,i} | \theta^{\prime})}
},
\label{eq:Bayesupdate}\\
\widehat{\bm{\psi}}_{k,i}(\theta)&=
\begin{cases}
\bm{\psi}_{k,i}(\theta),&\textnormal{for }\theta=\thetashare,\\
\displaystyle{\frac{1-\bm{\psi}_{k,i}(\thetashare)}{H-1}},&\textnormal{for }\theta\neq\thetashare,
\end{cases}
\label{eq:justone}\\
\bm{\mu}_{k,i}(\theta)&=
\frac{
	\exp\left\{
	\sum\limits_{\ell=1}^N a_{\ell k} \log \widehat{\bm{\psi}}_{\ell,i}(\theta)
	\right\}
}
{
		\sum\limits_{\theta^{\prime}\in\Theta} \exp\left\{\sum\limits_{\ell=1}^N a_{\ell k} \log \widehat{\bm{\psi}}_{\ell,i}(\theta^{\prime})
		\right\}
}.
\label{eq:combine}
\end{align}
In \eqref{eq:Bayesupdate}, the agents perform a Bayesian update to construct an intermediate belief vector $\bm{\psi}_{k,i}$. 
In order to obtain the final belief vector $\bm{\mu}_{k,i}$, we consider the usual log-linear combination rule seen in \eqref{eq:combine}. However, each agent $k$ shares only the component $\bm{\psi}_{k,i}(\thetashare)$. The receiving agents will then split the remaining mass $1-\bm{\psi}_{k,i}(\thetashare)$ uniformly across the values $\theta\neq\thetashare$. For this reason, the combination rule is applied to the modified belief vectors $\{\widehat{\bm{\psi}}_{\ell,i}\}_{\ell=1}^N$ in (3).
\vspace{-15pt}
\subsection{Self-aware partial approach}
The second approach consists in rewriting the combination step of the algorithm in such a way that agent $k$ combines its neighbors modified beliefs $\{\widehat{\bm{\psi}}_{\ell,i}\}^N_{\ell=1,\ell\neq k}$ with its own true belief $\bm{\psi}_{k,i}$:
\begin{equation}
\bm{\mu}_{k,i}(\theta)=
\frac{
	\exp\left\{
	a_{kk}\log\bm{\psi}_{k,i}(\theta)+\sum\limits_{\substack{\ell=1\\\ell\neq k}}^N a_{\ell k} \log \widehat{\bm{\psi}}_{\ell,i}(\theta)
	\right\}
}
{
	\sum\limits_{\theta^{\prime}\in\Theta} \exp\left\{a_{kk}\log\bm{\psi}_{k,i}(\theta^{\prime})+
		\sum\limits_{\substack{\ell=1\\\ell\neq k}}^Na_{\ell k}\log \widehat{\bm{\psi}}_{\ell,i}(\theta^{\prime})
		\right\}
}.
\label{eq:combine_nc}
\end{equation}
We can distinguish in \eqref{eq:combine_nc} two terms in its numerator: a first \emph{self-awareness} term and a second term combining the partially true beliefs from neighbors.
\vspace{-5pt}
\section{Convergence analysis}\label{sec:convan}
\vspace{-5pt}
\subsection{Partial approach}
	In this approach, agent $k$ chooses to combine its own modified belief vector $\widehat{\bm{\psi}}_{k,i}$ with its neighbors' equally modified beliefs. Lemma \ref{lemma:rate} states a result regarding the asymptotic rate of convergence while Theorem \ref{the:belcoll} establishes the convergence behavior of belief vectors, under different regimes. Note that $D_{\sf KL}[\cdot]$ refers to the Kullback-Leibler divergence. In the following, when computing KL divergences the argument $\xi$ of the pertinent distributions is omitted for ease of notation, and a.s. refers to almost sure convergence. Due to space constraints, proofs will be partially introduced.
	\begin{lemma}[Asymptotic rate of convergence]\label{lemma:rate}
		Under Assumptions \ref{ass:2} and \ref{ass:3}, we have that, for any $k$ and for all $\theta \neq \thetashare$:
		\begin{multline}
		\lim_{i\rightarrow\infty}\frac{1}{i}
		\log\frac{\bm{\mu}_{k,i}(\theta)}{\bm{\mu}_{k,i}(\thetashare)}\stackrel{\textnormal{a.s.}}{=}
		D_{\sf KL}[L(\thetatrue)||L(\thetashare)]\\- D_{\sf KL}[L(\thetatrue)||P(\thetashare^c)],\label{eq:lem1}
		\end{multline}
		where we define the average likelihood function distinct from $\theta$ as
		\begin{equation}
		P(\xi|\theta^{c})\triangleq\frac{1}{H-1}\sum_{\tau\neq\theta}L(\xi|\tau).
		\end{equation}
	\end{lemma}
	\begin{proof}[Proof of Lemma \ref{lemma:rate}]
		Note that \eqref{eq:combine} is equivalent to writing, for every pair $\theta, \theta'\in \Theta$:
		\begin{equation}
		\log \frac{\bm{\mu}_{k,i}(\theta)}{\bm{\mu}_{k,i}(\theta')} =\sum_{\ell=1}^{N}a_{\ell k}\log \frac{\widehat{\bm{\psi}}_{k,i}(\theta)}{\widehat{\bm{\psi}}_{k,i}(\theta')}.
		\label{eq:combine_2}
		\end{equation}
		We see from \eqref{eq:justone} that, for all $\theta,\theta^{\prime}\neq\thetashare$,
		\begin{equation}
		\log\frac{\bm{\mu}_{k,i}(\theta)}{\bm{\mu}_{k,i}(\theta^{\prime})}=0\Rightarrow
		\bm{\mu}_{k,i}(\theta)=\bm{\mu}_{k,i}(\theta^{\prime}).\label{eq:equalmu}
		\end{equation}
		Considering \eqref{eq:equalmu} and the recursion in \eqref{eq:combine_2} with $\theta'=\thetashare$ results in 
		\begin{equation}
		\log \frac{\bm{\mu}_{k,i}(\theta)}{\bm{\mu}_{k,i}(\thetashare)} =\sum_{\ell=1}^{N}a_{\ell k}\left[\log\frac{\bm{\mu}_{\ell,i-1}(\theta)}{\bm{\mu}_{\ell,i-1}(\thetashare)}+
		\log\frac{
			P(\bm{\xi}_{\ell,i}|\thetashare^{c})}
		{L(\bm{\xi}_{\ell,i} | \thetashare)}\right].\label{eq:recursionnosa}
		\end{equation}
		Multiplying by $1/i$, taking the limit over $i$ as it goes to infinity, and using limiting arguments as in \cite{nedic2017fast}, we get 
		\begin{align}
		\lim_{i\rightarrow\infty}\frac{1}{i}
		\log\frac{\bm{\mu}_{k,i}(\theta)}{\bm{\mu}_{k,i}(\thetashare)}\stackrel{\textnormal{a.s.}}{=}
		\underbrace{\sum_{\ell=1}^{N}v_\ell}_{=1}\,\mathbb{E}\left[\log\frac{
			P(\bm{\xi}_{\ell,i}|\thetashare^{c})}
		{L(\bm{\xi}_{\ell,i}| \thetashare)}\right],\label{eq:limthetathetax}
		\end{align}
		which establishes the desired conclusion.
	\end{proof}\vspace{-5pt}
	\begin{theorem}[Belief collapse]\label{the:belcoll}
		Under Assumptions \ref{ass:2} and \ref{ass:3}, for any $k$, we have that:
		
		\noindent If $\thetashare = \thetatrue$, 
			\begin{equation}
			\hspace{5pt}D_{\sf KL}\left[L(\thetatrue)||P(\thetatrue^c)\right]>0\Rightarrow\lim_{i\rightarrow\infty}\bm{\mu}_{k,i}(\thetatrue)\stackrel{\textnormal{a.s.}}{=}1.\label{eq:th1}
			\end{equation}
			
		\noindent If conversely $\thetashare\neq\theta_0$, 
			\begin{equation}
			\hspace{10pt}\frac{D_{\sf KL}[L(\thetatrue)||P(\thetashare^c)]}{D_{\sf KL}[L(\thetatrue)||L(\thetashare)]}>1\Rightarrow\lim_{i\rightarrow\infty}\bm{\mu}_{k,i}(\thetashare)\stackrel{\textnormal{a.s.}}{=}1,\label{eq:cond1}
			\end{equation}
			and
				\begin{equation}
				\frac{D_{\sf KL}[L(\thetatrue)||P(\thetashare^c)]}{D_{\sf KL}[L(\thetatrue)||L(\thetashare)]}<1\Rightarrow\lim_{i\rightarrow\infty}\bm{\mu}_{k,i}(\theta)\stackrel{\textnormal{a.s.}}{=}\frac{1}{H-1},\label{eq:cond2}
				\end{equation}
for all $\theta \neq \thetashare$.
	\end{theorem}
\begin{proof}[Proof of Theorem \ref{the:belcoll}]
	 When $\thetashare=\thetatrue$, if the inequality in \eqref{eq:th1} holds, the convergence rate from Lemma \ref{lemma:rate} will be strictly negative, which implies that, since $\bm{\mu}_{k,i}(\theta_0)$ is bounded by 1, for any $\theta\neq \thetatrue$, $\lim_{i\rightarrow\infty}\bm{\mu}_{k,i}(\theta)\stackrel{\textnormal{a.s.}}{=}0$, resulting in \eqref{eq:th1}. When $\thetashare\neq\thetatrue$, the convergence behavior will depend on the sign of the RHS of \eqref{eq:lem1}. If the inequality in \eqref{eq:cond1} holds, the sign is negative and we get the convergence in \eqref{eq:cond1}. If the inequality in \eqref{eq:cond2} holds, the sign is positive, which implies that $\bm{\mu}_{k,i}(\thetashare)$ goes to zero, and all $\bm{\mu}_{k,i}(\theta)$ are equal in view of \eqref{eq:equalmu}.
\end{proof}\vspace{-5pt}
\setcounter{figure}{1}
\begin{figure*}
	\includegraphics[width=\textwidth]{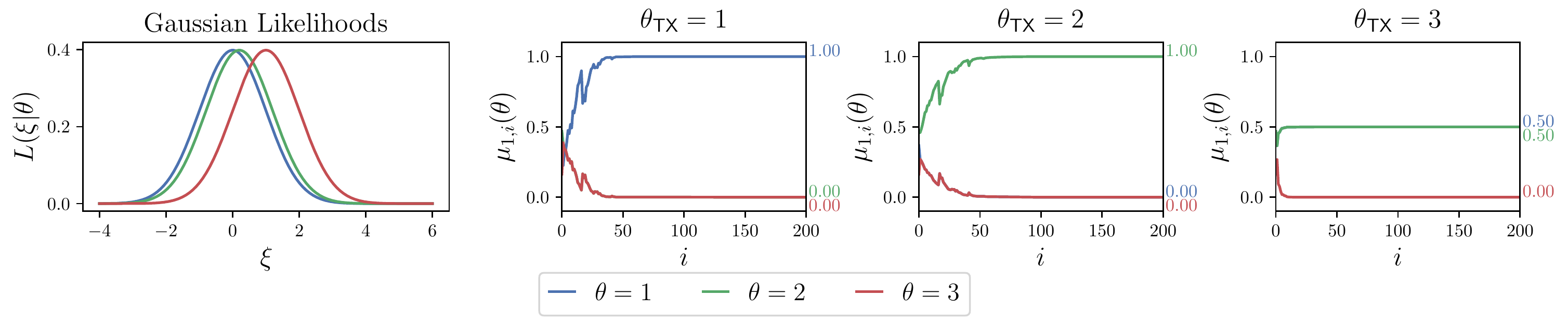}
	\caption{\emph{Partial approach}: Belief evolution for agent 1 for different choices of $\thetashare$.}
	\label{fig:na}	
\end{figure*}
	Theorem \ref{the:belcoll} reveals that when the transmitted hypothesis $\thetashare$ is equal to the true one, all agents are able to recover the truth. Moreover, in the alternative scenario, two situations can be produced: either all agents will converge to a belief concentrated at the hypothesis $\thetashare$, or they will converge to a uniform splitting across the non-transmitted hypotheses. This behavior admits a very useful interpretation. When \eqref{eq:cond1} is verified, the transmitted hypothesis is more easily confounded with the true one (since their KL divergence is small). In comparison, when \eqref{eq:cond2} is verified, there is no sufficient evidence for choosing $\thetashare$, but since in view of \eqref{eq:equalmu} all the unshared beliefs remain equal, the only meaningful choice of the agents is to split equally their beliefs across the remaining $H-1$ hypotheses.
\vspace{-5pt}
\subsection{Self-aware partial approach}
Now, we consider that agent $k$ will combine its true belief vector $\bm{\psi}_{k,i}$ with its neighbors' modified belief vectors $\{\widehat{\bm{\psi}}_{\ell,i}\}^N_{\ell=1,\ell\neq k}$. When $\thetashare=\thetatrue$, we have the result stated in Theorem \ref{th:truthselfaware}. For this approach, proofs are omitted due to space constraints.
\begin{theorem}[Truth learning with self-awareness]\label{th:truthselfaware}
	Under Assumptions \ref{ass:2} and \ref{ass:3}, when $\thetashare=\thetatrue$, if $D_{\sf KL}\left[ L(\thetatrue)||P_q(\thetatrue^c)\right]\neq 0$ for any pmf $q(\cdot)$ over $\Theta\backslash\thetatrue$ such that we define
	 \begin{equation}
	P_q(\xi|\thetatrue^c)\triangleq \sum_{\tau\neq \thetatrue}q(\tau)L(\xi|\tau),
	\end{equation}
	then for any $k$, we have 
	\begin{equation}
	\lim_{i \rightarrow \infty}\bm{\mu}_{k,i}(\theta_0)\stackrel{\textnormal{a.s.}}{=}1.
	\end{equation}\qed
\end{theorem}\vspace{-5pt}
On the other hand, when $\thetashare\neq \thetatrue$, two situations may arise: the belief concerning the transmitted hypothesis will collapse either to zero or to one, under the sufficient conditions given in Lemmas \ref{lem:sufcond0} and \ref{lem:sufcond2}, respectively. 
\begin{lemma}[Alternative behavior with self-awareness]\label{lem:sufcond0}
	Under Assumption \ref{ass:2} and \ref{ass:3}, when $\thetashare\neq \theta_0$, for any $k$, if 
	\begin{equation}
	D_{\sf KL}[L(\theta_0)||
	L(\theta_{\sf TX})] >\frac{\alpha}{H-1}\sum_{\tau\neq \thetashare}D_{\sf KL}[L(\theta_0)||L(\tau) ],
	\end{equation}
	then
	\begin{equation}
	\lim_{i \rightarrow \infty}\bm{\mu}_{k,i}(\thetashare)\stackrel{\textnormal{a.s.}}{=} 0,
	\end{equation}
	where we defined $\displaystyle{\alpha\triangleq \sum_{\ell =1}^{N}v_\ell\sum_{\substack{n=1\\n\neq \ell}}^N\frac{a_{n \ell}}{1-a_{n n}}}$.\qed
\end{lemma}
Under the sufficient condition in Lemma \ref{lem:sufcond0}, we discover that the belief concerning $\thetashare$ collapses to zero with probability one. As for the remaining elements, it can be shown that they will exhibit an oscillatory behavior according to the recursion
\begin{equation}
\log\frac{\bm{\mu}_{k,i}(\theta)}{\bm{\mu}_{k,i}(\theta')}=a_{kk}\log\frac{\bm{\mu}_{k,i-1}(\theta)}{\bm{\mu}_{k,i-1}(\theta')}+a_{kk}\log\frac{L(\bm{\xi}_{k,i}|\theta)}{L(\bm{\xi}_{k,i}|\theta')}\label{eq:recursion}
\end{equation}
for $\theta, \theta'\neq \thetashare$. This scenario differs from the partial approach, where the above recursion is given instead by \eqref{eq:equalmu}, where the beliefs for non-transmitted hypotheses evolve equally.

Before presenting Lemma \ref{lem:sufcond2}, we introduce an extra assumption on the boundedness of the likelihood functions.
\begin{assumption}[Bounded likelihoods]\label{ass:5}
	Let there be $M>0$, such that, for all $\xi\in\mathcal{X}$ and for any pair $\theta,\theta^{\prime}\neq\thetashare$, we have:
\begin{equation}
-M\leq \log \frac{L(\xi|\theta)}{L(\xi|\theta^{\prime})}\leq M.\label{eq:boundl}
\end{equation}\QED
\end{assumption}\vspace{-5pt}

\begin{lemma}[Mislearning with self-awareness]\label{lem:sufcond2}
	Under Assumptions \ref{ass:2}, \ref{ass:3}  and \ref{ass:5}, for any $k$, when $\thetashare\neq \theta_0$, if
	\begin{multline}
	D_{\sf KL}[L(\thetatrue)||P(\thetashare^c)] > D_{\sf KL}[L(\thetatrue)||L(\thetashare)]
	\\+
	M\sum_{\ell =1}^Nv_{\ell}\sum_{\substack{n=1\\n\neq \ell}}^N a_{n \ell}\frac{a_{n n}}{1-a_{nn}},\label{eq:loose}
	\end{multline}
	then
	\begin{equation}
	\lim_{i\rightarrow\infty} \bm{\mu}_{k,i}(\thetashare)\stackrel{\textnormal{a.s.}}{=}1.
	\end{equation}\qed
\end{lemma}\vspace{-5pt}
Lemma \ref{lem:sufcond2} reveals that, for some sufficiently small choice of self-awareness weights $a_{n n}$ for every $n=1,2,\dots,N$, as long as $D_{\sf KL}[L(\thetatrue)||P(\thetashare^c)] > D_{\sf KL}[L(\thetatrue)||L(\thetashare)]$ holds, the belief concerning the transmitted hypothesis $\thetashare$ will converge to one, while all others will go to zero. The sufficient condition in \eqref{eq:loose} is expected to be not tight because Lemma \ref{lem:sufcond2} relies on the bounds in \eqref{eq:boundl}, which are usually not tight. This condition is nonetheless useful since it highlights the possibility that mislearning occurs.
\vspace{-5pt}
\section{Simulation results}\label{sec:simres}
\vspace{-5pt}
Consider a network, whose topology can be seen in Fig. \ref{fig:network} (self-loops are present at every node and are not displayed) and whose combination matrix $A$ is determined using a parametrized averaging rule \cite{sayed2014diffusion}:
\begin{equation}
a_{\ell k} = \begin{cases}
	 \displaystyle{\lambda},&\text{if }\ell = k\\
	\displaystyle{\frac{1-\lambda}{n_k-1}}, &\text{if }\ell \neq k \text{ and } \ell \in \mathcal{N}_k\\
	0,&\text{ otherwise},
\end{cases}\label{eq:matrix}
\end{equation}
where $\lambda$ is a hyperparameter used to tune self-awareness and $n_k$ is the degree of agent $k$ (including $k$ itself).
The set of hypotheses is given by $\Theta=\{1,2,3\}$, and the true state of nature corresponds to the first hypothesis, i.e., $\thetatrue=1$. 
\vspace{-10pt}
\setcounter{figure}{0}
\begin{figure}[htb]
	\centering
	\includegraphics[width=.8\linewidth]{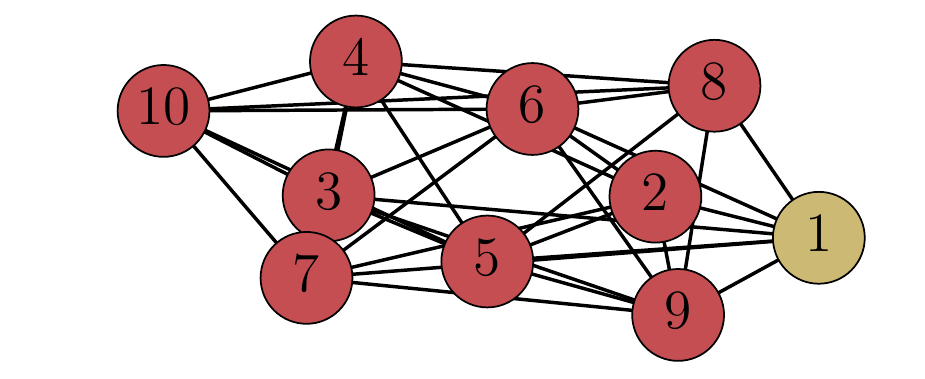}
	\caption{Network topology, where agent 1 is highlighted. }
	\label{fig:network}
\end{figure}
\setcounter{figure}{2}
\begin{figure*}
	\includegraphics[width=\textwidth]{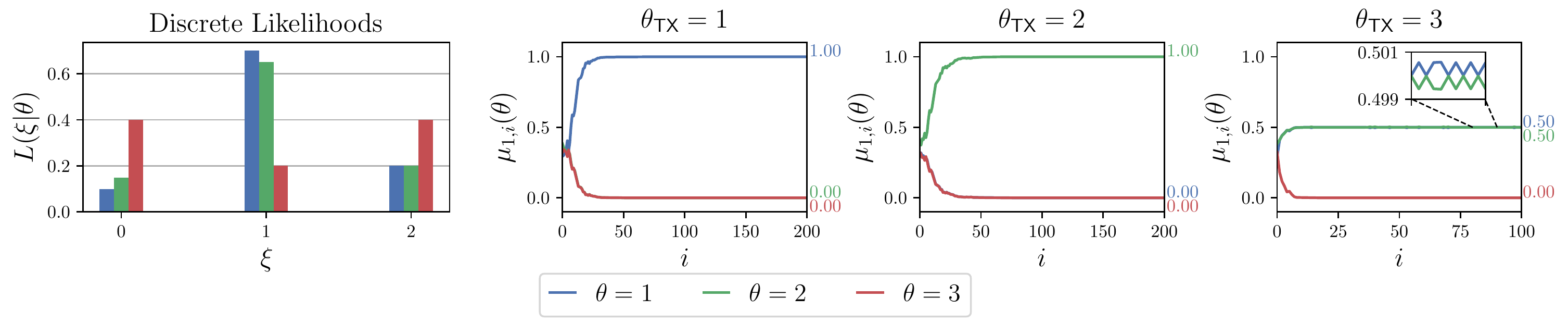}
	\caption{\emph{Self-aware partial approach:} Belief evolution for agent 1 for different choices of $\thetashare$.}\vspace{-10pt}
	\label{fig:sa}	
\end{figure*}
\vspace{-10pt}
\subsection{Partial approach}
We first consider the partial approach, with hyperparameter $\lambda = 0.5$. Assume the same family of unit-variance Gaussian likelihood functions across all agents, whose means are chosen as $0$, $0.2$ and $1$ for $\theta=1$, $\theta=2$ and $\theta=3$ respectively, as seen in the left part of  Fig. \ref{fig:na}. We will focus on the behavior of agent $1$ for illustration purposes. As predicted by Theorem \ref{the:belcoll}, when $\thetashare=\thetatrue=1$, we observe that the agent's belief collapses with full confidence to the true hypothesis as seen in Fig. \ref{fig:na}. For the other cases, when $\thetashare=2$, since $D_{\sf KL}[L(\thetatrue)||L(\thetashare)]- D_{\sf KL}[L(\thetatrue)||P(\thetashare^c)]=-0.091<0$, the agent's belief vector converges with full confidence to the shared hypothesis. Alternatively, when $\thetashare=3$, since $D_{\sf KL}[L(\thetatrue)||L(\thetashare)]- D_{\sf KL}[L(\thetatrue)||P(\thetashare^c)]=0.494>0$, the agent's belief is equally split among all hypotheses different than the shared one.
\vspace{-5pt}
\subsection{Self-aware partial approach}
For this approach, we consider $\lambda = 0.03$ and a family of likelihood functions with discrete support over $\xi$, i.e. $\xi\in\{0,1,2\}$, as seen in the left panel of Fig. \ref{fig:sa}. In this scenario, Assumption \ref{ass:5} is satisfied, and, more particularly, the constant $M$ from \eqref{eq:boundl} is equal to $3.5$.
\begin{remark}[Parametric average combination matrix]\label{rem:2}
	If matrix $A$ is a left-stochastic matrix, whose weights follow an averaging rule such as in \eqref{eq:matrix}, then $\alpha =1$.\QED
\end{remark}
When $\thetashare = \thetatrue =1$, illustrating Theorem \ref{th:truthselfaware}, we observe almost sure convergence to the true hypothesis, as seen in Fig. \ref{fig:sa}. Some interesting phenomena arise from observing the evolution of beliefs when $\thetashare \neq \thetatrue$. 

When $\thetashare = 2$, $D_{\sf KL}[L(\thetatrue)||P(\thetashare^c)] - D_{\sf KL}[L(\thetatrue)||L(\thetashare)]
+
M\sum_{\ell =1}^Nv_{\ell}\sum_{n=1,n\neq \ell}^N a_{n \ell}\frac{a_{n n}}{1-a_{nn}} = 0.02>0$ and therefore Lemma \ref{lem:sufcond2} is satisfied. In simulation, for larger values of $\lambda$, even though the previous condition no longer holds, we still observe that $\lim_{i \rightarrow \infty}\bm{\mu}_{k,i}(\thetashare)=1$, which is due to the fact that Lemma \ref{lem:sufcond2} only establishes a sufficient condition. Another interesting observation from the simulation is that when $\lambda$ is sufficiently large, i.e., the self-awareness term is dominant, we are able to recover the alternative behavior for which $\lim_{i \rightarrow \infty}\bm{\mu}_{k,i}(\thetashare)=0$.

When $\thetashare = 3$, we observe that $D_{\sf KL}[L(\theta_0)|| L(\theta_{\sf TX})] -\frac{1}{H-1}\sum_{\tau\neq \thetashare}D_{\sf KL}[L(\theta_0)||L(\tau) ]=0.094>0$
and therefore, considering Remark \ref{rem:2}, Lemma \ref{lem:sufcond0} holds and no choice of $\lambda$ results in a change of regime, i.e., $\bm{\mu}_{k,i}(\thetashare)$ will always converge to zero almost surely. This behavior is seen in Fig. \ref{fig:na}, where the beliefs concerning hypotheses $\theta =1$ and $\theta =2$ oscillate indefinitely, according to recursion \eqref{eq:recursion}.

In this case, the self-awareness parameter $\lambda$ will influence the behavior of the beliefs concerning the non-transmitted hypotheses, more specifically, it will influence the amplitude of their oscillatory behavior. In Fig. \ref{fig:oscill}, we see how increasing $\lambda$ makes the oscillatory behavior of \eqref{eq:recursion} increase.
\begin{figure}[htb]
	\centering
	\includegraphics[width=\linewidth]{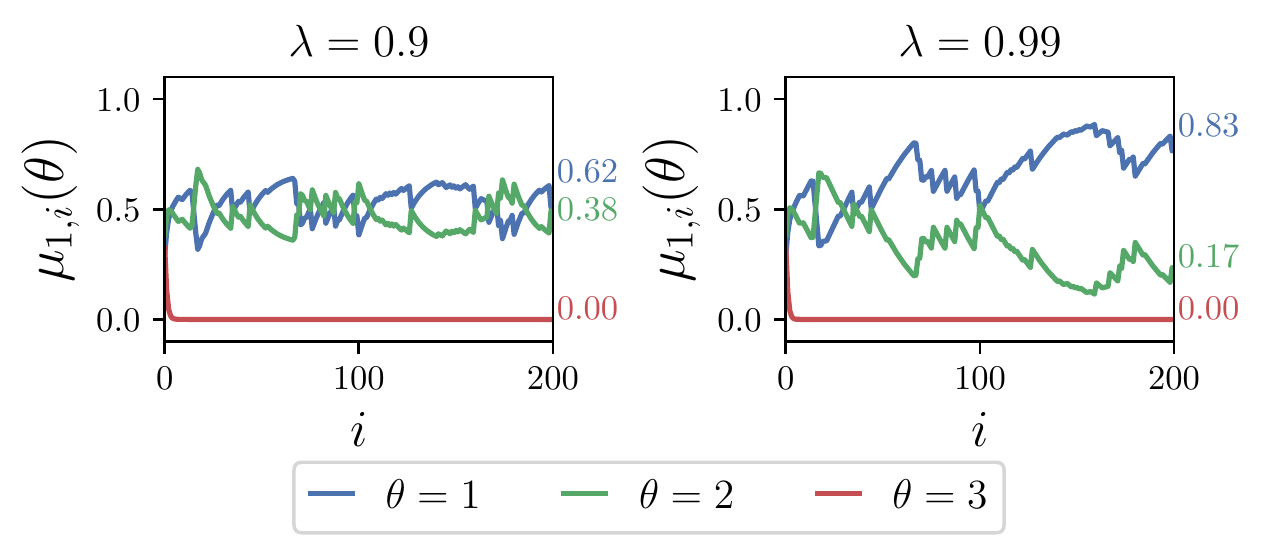}
	\caption{Belief evolution for agent 1 with $\thetashare=2$. \emph{Left:} For $\lambda=0.9$. \emph{Right:} For $\lambda=0.99$.}\label{fig:oscill}\vspace{-5pt}
\end{figure} 
Besides, since the problem is identifiable at every agent, a growing self-awareness parameter brings the agents closer to the truth. Note that in the limit, when $\lambda\rightarrow 1$, we obtain a non-cooperative learning framework.
\vspace{-5pt}
\subsection{Concluding remarks and a third approach}
In this work, we discussed two social learning algorithms that incorporate the notion of partial information sharing. Both algorithms rely on the fact that agents choose deterministically to share one hypothesis from among the available ones. Another interesting approach would be to consider instead the following intermediate step: 

\begin{equation}
\widehat{\bm{\psi}}_{k,i}(\theta)=
\begin{cases}
\bm{\psi}_{k,i}(\theta),&\textnormal{for }\theta=\bm{\theta}^{k,i}_{\sf MAX},\\
\displaystyle{\frac{1-\bm{\psi}_{k,i}(\bm{\theta}^{k,i}_{\sf MAX})}{H-1}},&\textnormal{for }\theta\neq\bm{\theta}^{k,i}_{\sf MAX},
\end{cases}\label{eq:thetamax}
\end{equation} 
where we defined the random variable $\bm{\theta}^{k,i}_{\sf MAX}\triangleq \mathrm{argmax}_{\theta}(\bm{\psi}_{k,i}(\theta))$ as the most relevant hypothesis for agent $k$ at instant $i$. The behavior of the algorithm without self-awareness, considering \eqref{eq:thetamax} as intermediate step can be seen in Fig. \ref{fig:max} for uniform initial beliefs and for random initializations over 100 Monte Carlo runs.
\vspace{-5pt}
\begin{figure}[htb]
	\centering
	\includegraphics[width=\linewidth]{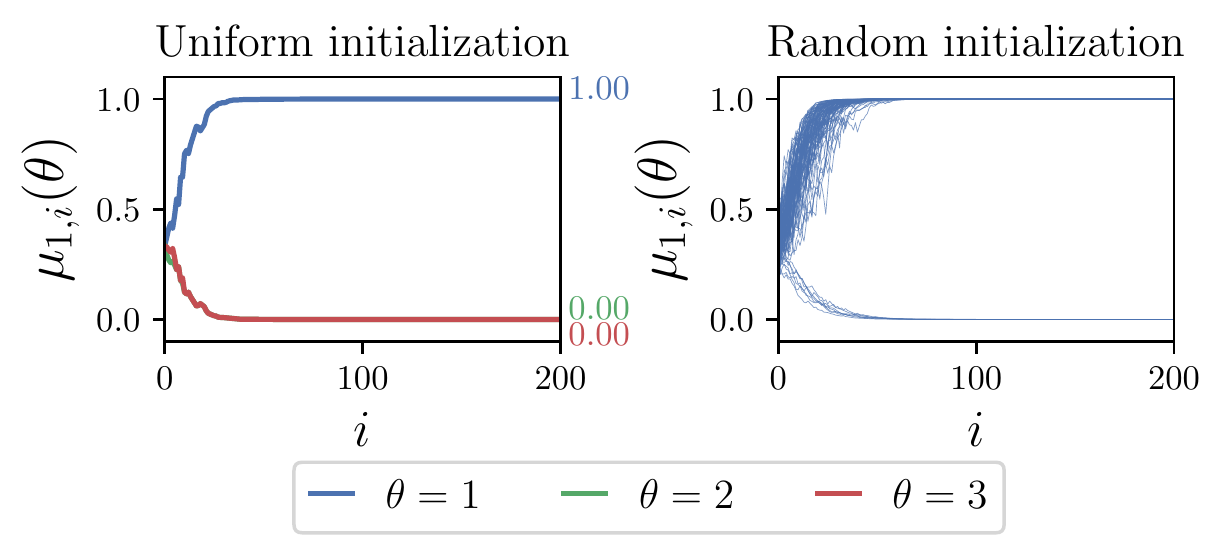}
	\caption{Social learning with most relevant hypothesis sharing. \emph{Left:} Uniform initialization. \emph{Right:} Random initialization with 100 Monte Carlo experiments.} \label{fig:max}
\end{figure}
\vspace{-5pt}

Observing the convergence results in Fig. \ref{fig:max}, the following remarkable behavior emerges. First (left of Fig. \ref{fig:max}), under uniform initialization, partial information with max-belief sharing can be enough to learn properly the true hypothesis. 
Moreover, under random initialization (right of Fig. \ref{fig:max}), sometimes mislearning arises. We have verified on the data that, notably, mislearning is associated to realizations with poor initialization. In these realizations, $\bm{\theta}^{k,i}_{\sf MAX}$ tends to be equal to $\theta=2$ in the first iterations, whose behavior matches the case seen in \eqref{eq:cond1} of Theorem \ref{the:belcoll}. 

\vfill\pagebreak
\balance

\bibliographystyle{IEEEbib}
\bibliography{../biblio}{}

\end{document}